\documentclass[10pt, conference, compsocconf]{IEEEtran}
\ifCLASSINFOpdf
\else
\fi
\usepackage{caption}
\usepackage{subfigure}
\usepackage{amsthm}
\newtheorem{theorem}{Theorem}

\newtheorem{lemma}{Lemma}
\newtheorem{assumption}{Assumption}
\usepackage{amsmath}
\usepackage{mathtools}
\usepackage{url}
\usepackage{algorithm}
\usepackage{algorithmic}
\usepackage{dsfont}
\usepackage{bm}

\usepackage{cite}
\usepackage{amsmath,amssymb,amsfonts}
\usepackage{algorithmic}
\usepackage{graphicx}
\usepackage{textcomp}
\usepackage{xcolor}

\begin{document}
%



%

\title{ Towards  Sustainable  Satellite Edge Computing}
\author{\IEEEauthorblockN{Qing Li, Shangguang Wang, Xiao Ma, Ao Zhou, Fangchun Yang}
\IEEEauthorblockA{State Key Laboratory of Networking and Switching Technology\\
Beijing University of Posts and Telecommunications\\
Beijing, China\\
\{q\_li;sgwang;maxiao18;aozhou;fcyang\}@bupt.edu.cn}}

\maketitle

\begin{abstract}
Recently, Low Earth Orbit (LEO)  satellites experience rapid development and satellite edge computing emerges to address the limitation of bent-pipe architecture in existing satellite systems. Introducing energy-consuming computing components in satellite edge computing increases the depth of battery discharge. This will shorten batteries' life and influences the satellites' operation in orbit. In this paper, we aim to extend batteries' life by minimizing the  depth of discharge for Earth observation missions.   Facing the challenges of wireless uncertainty and energy harvesting dynamics, our work develops 
an online energy scheduling algorithm within  an online convex optimization framework. Our algorithm achieves sub-linear regret and the constraint violation asymptotically  approaches zero. Simulation results show that our algorithm can reduce the depth of discharge significantly. 

\end{abstract}
\IEEEpeerreviewmaketitle
\section{Introduction}
Recently,  Low Earth Orbit (LEO)  satellites  experience  rapid development  due to the reduced  cost of both manufacturing and launching.  
Existing  satellite systems  operate under a bent-pipe architecture \cite{larson1992space}, where ground stations send control commands to orbits and satellites reply with raw data.  This architecture relies heavily on the satellite-ground communication, which has limitations  of high downlink latency, intermittent availability, and link unreliability  \cite{10.1145/3373376.3378473}.   Satellite edge computing equips satellites with computing resources and supports in-orbit data processing. Thus, it can reduce the downlink transmission load and  provide scalability benefits  when LEO satellite constellations scale up \cite{10.1145/3373376.3378473}.

Introducing energy-consuming computing components in satellite edge computing brings a huge burden to energy systems, of which the capacities are inherently constrained due to  strict volume and weight constraints \cite{mysatellite}. Most LEO  satellites have solar cells installed on their surface to harvest solar energy  and store energy (usually in batteries) to keep  functioning during the eclipse. Both  energy harvesting and storage are constrained  by physical sizes.  For example, a common category of LEO satellites, CubeSats, are among [1, 15] kg in weight and of up to 12U in volume (1U = 10 cm × 10 cm × 10 cm) \cite{Davoli2019SmallSA}.  The harvested power of CubSats  ranges in  [1, 7] watts because of the limited area of the solar arrays \cite{Davoli2019SmallSA}.  Besides,  the capacities of batteries  in CubeSats range in dozens of watt-hour\cite{mysatellite}. 
Introducing the computing components increases the depth of discharge (i.e., the amount of discharge energy) in each eclipse period.   This  will greatly shorten  batteries' life, which further influences the satellites' operation in orbit \cite{app8112078}.  Therefore, it is timely and important to extend  batteries’ life.

  Extensive previous works focus on  energy scheduling in ground energy harvesting systems  involving many issues under different network contexts \cite{9205292, fraternali2020aces, huang2021adaptive, 9217302}. Applying these works directly to satellite systems is appealing, but they cannot address the  challenges brought by  high-speed satellite movement. Moreover, none of them attempts to extend  batteries' life. One related work \cite{8977503} focuses on extending batteries' life  by transmission power control in LEO satellite networks. In our scenario, introducing edge computing to LEO satellite networks makes the problem more complex as it requires energy  coordination between different energy-consuming components.

  It is non-trivial to extend the batteries’ life of LEO   satellites due to the following challenges. The first challenge is how to optimize the depth of discharge in the unstable wireless environment. Each satellite orbits the Earth every $ \sim$100 minutes, traveling at $ \sim$27,000 kmph \cite{bhattacherjee2019network}.  This high-speed movement of satellites creates high churn in satellite-ground links. Further, the bitrates of satellite-ground links are unpredictable due to the uncertainty  of wireless environment such as weather conditions. The second challenge is how to adapt to energy harvesting dynamics. Typical low Earth orbits expose satellites to the Sun for about 66\% of each $\sim$100 minutes orbit period \cite{Davoli2019SmallSA}.  The periodical satellite movement incurs significant energy harvesting dynamics when satellites show up in light and eclipse alternatively.   
  
In this paper,  we seek to extend the batteries’ life in satellite edge computing by reducing the depth of discharge. We consider Earth observation missions, which consist of three energy-consuming processes: sensing, computing, and communication.  We exploit the pattern information brought by periodical satellite movement and propose a novel  optimal pattern-aware benchmark that generalizes state-of-the-art.
Given the uncertainty of the wireless environment and the dynamics of energy harvesting,  we propose a pattern-aware online energy scheduling algorithm within the online convex optimization framework.   Our algorithm achieves  sub-linear regret (compared with the  optimal pattern-aware benchmark) and  the constraint violation asymptotically  approaches zero. Simulation results show that our algorithm adapts to  the energy harvesting dynamics and  reduces  the depth of  discharge significantly. 

\section{Related Work}
\textbf{Energy scheduling in ground energy harvesting systems}.
  Energy harvesting systems  have been  investigated extensively on the ground.  Most existing workS learn energy scheduling strategies online by reinforcement learning method to address the uncertainty in the energy harvesting process. 
    Ortiz \textit{et al.} learn a distributed energy allocation for both a transmitter and a relay with only partially observable system states \cite{9205292}.  
    Fraternali  \textit{et al.} aim to maximize the sensing quality of energy harvesting sensors for periodic and event-driven indoor sensing with available energy \cite{fraternali2020aces}. 
    Huang \textit{et al.}  propose an  adaptive processor frequency adjustment algorithm  to plan  the energy usage of energy harvesting edge servers \cite{huang2021adaptive}.
  Hatami \textit{et al.} control  sensors status update  to minimize the energy cost considering the freshness requirement of the sensing information \cite{9217302}.  
 These works cannot be applied  directly to the satellite edge computing systems for two reasons. First, they cannot address the challenges brought by  the unique satellite movement, i.e., intermittent link availability.  Second, they do not investigate the problem of extending the batteries' life.

\textbf{Satellite edge computing}.
 Satellite edge computing  is still in its infancy. Most  existing works  focus  on  the space-air-ground network.  Boero \textit{et al.} \cite{8473415}, Giambene \textit{et al.} \cite{8473417}, and  Shi  \textit{et al.} \cite{8610425}   design the space-air-ground network architecture based on SDN and NFV  technologies. 
Tang \textit{et al.} \cite{9155485} manage  resources for SDN-based satellite-terrestrial networks in an on-demand way.  Chen \textit{et al.} propose the time-varying resource graph  to model resources in space-terrestrial integrated networks  \cite{myinfocom1}.  They  also  investigate how to dynamically place and
assign controllers in LEO satellite networks to adapt to satellite mobility and traffic load fluctuation \cite{myinfocom2}.
 Besides, a few works focus on satellite edge computing  frameworks \cite{10.1145/3422604.3425937}.   Denby \textit{et al.}  \cite{10.1145/3373376.3378473} propose to support computing in nano-satellite constellations to address existing ben-pipe architecture limitations. 
  Tsuchida \textit{et al.} \cite{8977503} also focus on extending battery life in transmission power control  problems for LEO satellite networks. Introducing edge computing to LEO satellite networks makes energy scheduling more complex as it requires energy  coordination between different energy-consuming components.


\section{System Model}
 As  shown in Figure~\ref{tu1}, we consider  an electrical power subsystem  model  in satellites \cite{8289325}. It encompasses efficient and reliable energy generation, storage, and  distribution to various on-board subsystems.
 The  electrical power subsystem consists of  a solar module, a battery module, a payload module, and an energy scheduler. 
 The solar module is responsible for  energy harvesting and transfer.  
 Solar energy is the dominant energy source of  LEO satellites, and about 85\% of nano-satellites are equipped with solar panels\cite{mysatellite}. 
 The energy storage module provides energy
supply when the solar energy is unavailable during on-orbit operations.
 The payload module regulates the energy to the other subsystems.
 On-board subsystems consume electrical  energy to maintain satellite orbit motion or perform various missions. 
 The energy scheduler decides how to allocate electrical energy to these subsystems.  
 In this paper, we consider an energy scheduling problem for an Earth observation mission. 
To optimize the energy allocation among sensing,  computing, and communication process, we abstract the energy harvesting,  storage, and distribution as an energy queue model as shown in Figure \ref{tu2}.  We divide time into slots indexed by $t$  with duration $T_\mathrm{d}$. 
\subsection{Energy Harvesting and Storage Model} \label{sec:Energy Harvesting and Storage Model}
\begin{figure}[t]
\begin {center}
\centerline{\includegraphics[width=0.45\textwidth]{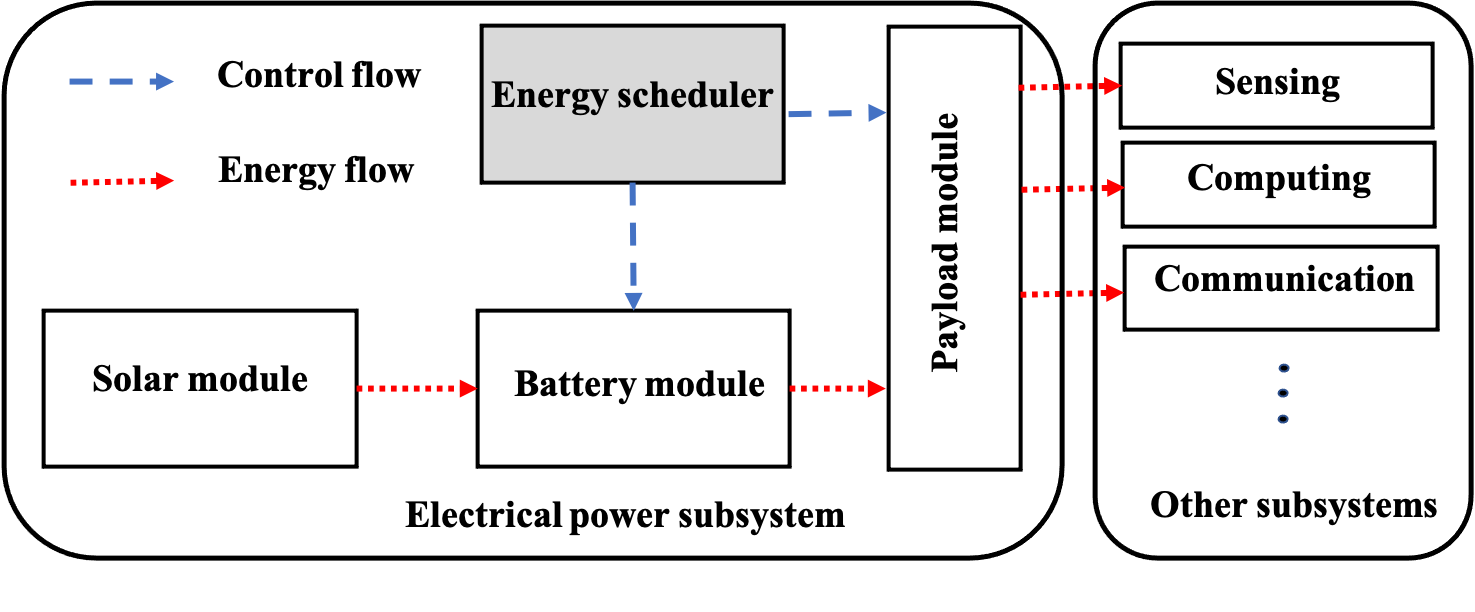}}
\caption{ System model.}
\label{tu1}
\end {center}
\vspace{-6ex}
\end{figure}
 Given the solar panel hardware, the  energy harvesting rate $r_e(t)$  is dominated by  two main factors: the light available to a satellite solar array  and the  projected surface area of the panels exposed to the Sun. 
 The first factor varies with the inverse square of the distance from the Sun.
 The second factor varies with  the angle between the solar panel and the Sun.
 We model the harvested energy in each time slot $t$ as 
\begin{align} \label{Equation:1}
E_{\mathrm{total}}(t)=r_e(t)T_{\mathrm{d}}.
\end{align}

The harvested energy varies intensely when satellites are in different light conditions.
In the light, the harvested energy  provides power to the energy-consuming subsystems directly  and charges the batteries.  
In the eclipse, batteries provide energy  to subsystems.
  The  dynamics of energy buffered in the batteries can be modeled as an energy queue 
\begin{align}\label{Equation:2}
E(t)=\min\{\max(E(t-1)+E_{\mathrm{in}}(t)-E_{\mathrm{out}}(t),0), E_{\max}\},
\end{align}
where $E_{\mathrm{in }}(t)$ is the amount of charge energy,  $E_{\mathrm{out}}(t)$ is the amount of discharge energy,   and $E_{\max}$ is the maximal battery capacity.
\subsection{Energy Consumption Model }  
The energy consumption  can be divided into two categories: the mission level energy consumption (including sensing energy $E_\mathrm{sen}(t)$, computing energy $E_\mathrm{cmp}(t)$, and communication energy $E_\mathrm{com}(t)$) and the energy to perform the fundamental operations of the satellites denoted by $E_a(t)$.
 The total energy consumption satisfies
\begin{align}\label{Equation:4}
 E_\mathrm{sen}(t)+E_\mathrm{cmp}(t)+ E_\mathrm{com}(t)+E_a(t) \leq\\
\nonumber
E_{\mathrm{out}}(t)+ E_\mathrm{total}(t)-E_\mathrm{in}(t).
\end{align}
\subsubsection{Sensing Energy Model}
The satellite collects the images of the Earth while moving along its orbit. The camera's  frame rate $f_\mathrm{sen}(t)$  (in frame/s)  is adjustable to meet different missions demands. 
The sensing energy is 
\begin{align} \label{Equation:5}
E_\mathrm{sen}(t)=P_\mathrm{sen}T_{\mathrm{d}}, 
\end{align}
where $P_\mathrm{sen}$ is the power of the camera.

\begin{figure}[t]
\begin {center}
\centerline{\includegraphics[width=0.45\textwidth]{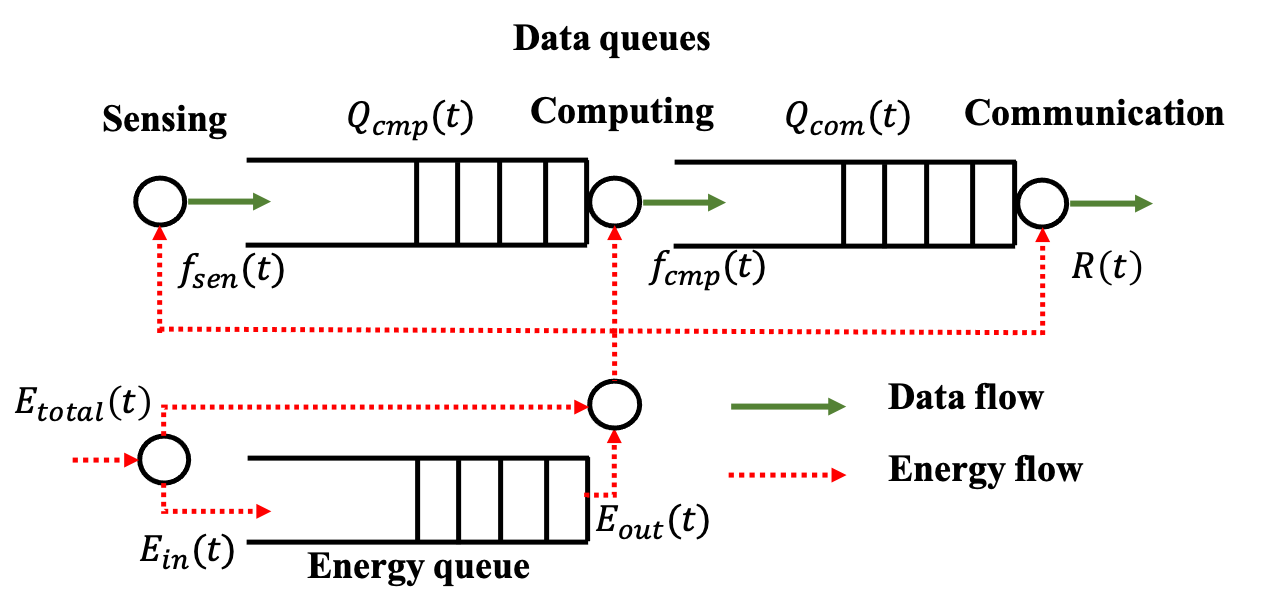}}
\caption{ Queue model.}
\label{tu2}
\end {center}
\vspace{-6ex}
\end{figure}
\subsubsection{Computing Energy Model}
The sensed data (e.g., Earth imagery) can be hundreds of  Gigabytes  and  its quality is fundamentally limited by the on-board cameras and orbit altitude \cite{10.1145/3373376.3378473}.  It is unnecessary to transmit all the raw data to the ground.  Besides, the satellite-ground downlink is not always available and the link bitrate is affected by many factors, e.g., orbit parameters, ground station capability, and location. 
These necessitate in-orbit data computing, which can significantly relieve the satellite-ground link pressure. A typical on-board computing example is to identify images of interest and separate them from raw data by CNN-based image classification, objective detection, or  any other computation. The computing system adopts the dynamic voltage and frequency scaling technique \cite{le2010dynamic} to adjust its CPU frequency denoted by $f_\mathrm{cmp}(t)$ (in cycle/s). Let $0\leq f_\mathrm{cmp}(t) \leq f_\mathrm{cmp}^{\max}(t)$, where $f_\mathrm{cmp}^{\max}(t)$ is the maximal CPU frequency. 
The computation energy consumption is 
\begin{align} \label{Equation:6}
E_\mathrm{cmp}(t)=a f_\mathrm{cmp}^3(t)T_{\mathrm{d}},
\end{align}
where $a$ is the effective capacitance coefficient of computing hardware.
\subsubsection{Communication Energy Model}
The  satellite-ground  connection is intermittent because the satellite moves fast with respect to  ground stations and the downlink session can only last for less than ten minutes in one single pass \cite{10.1145/3373376.3378473}. 
The satellite  stores the processed data before connecting to a ground station.
In each time slot $t$, the satellite-ground connection $I_c(t)$ is known as a prior, where $I_c(t)=1$ means that  the satellite-ground connection is available, and vice versa. We model the transmission rate  $R(t)$ (in bit/s)  as
\begin{align}\label{Equation:7}
R(t)=B\log_2(1+\frac{P_\mathrm{com}(t)h(t)}{N_0}),
\end{align}
where $P_\mathrm{com}(t)$ is the transmit power of the satellite, $N_0$ is the received noise power, and $h(t)$ is the channel gain.  Let $0\leq R(t) \leq R^{\max}I_c(t)$, where $R^{\max}$ is the maximal transmission rate. The communication energy consumption is 
\begin{align}\label{Equation:8}
E_\mathrm{com}(t)=P_\mathrm{com}(t)T_{\mathrm{d}}.
\end{align}

\subsection{Data Queue Model}
 To describe the data buffers dynamics, we construct two data queues: a waiting-for-computing data queue $Q_\mathrm{cmp}(t)$ between sensing  and computing processes and a waiting-for-transmitting  data queue $Q_\mathrm{com}(t)$ between computing and communication processes.  
 The arrival rate of the queue  is the  sensed data amount  and the departure rate is the processed data amount in the current slot.
   We can update the waiting-for-computing data queue $Q_\mathrm{cmp}(t)$ as
\begin{align}\label{Equation:9}
Q_\mathrm{cmp}(t)=\max\{Q_\mathrm{cmp}(t-1)+A_s(t)-\kappa f_\mathrm{cmp}(t)T_{\mathrm{d}},0\},
\end{align}
where each CPU cycle executes $\kappa$ bits of data known as a prior,  $A_s(t)=D f_\mathrm{sen}(t)T_{\mathrm{d}}$  is the amount of sensed data, and   $D$ is the data size of each image frame.

After the mission completion, all sensed information is required to be downloaded to the ground for further use. 
  Hence, we should have $Q_\mathrm{cmp}(T+1)=0$ with the initial queue length $Q_\mathrm{cmp}(1)\geq0$ by the long-term constraint
\begin{align}\label{Equation:11}
\sum_{t=1}^{T}(Df_\mathrm{sen}(t)-\kappa f_\mathrm{cmp}(t))\leq 0.
\end{align}
The  waiting-for-transmitting data queue $Q_\mathrm{com}(t)$ buffers  results of interest that need to be transmitted to the ground.
 Similarly, we can update the waiting-for-transmitting data queue as
\begin{align}\label{Equation:12}
Q_\mathrm{com}(t)=\max\{Q_\mathrm{com}(t-1)+\rho A_p(t)-R(t)T_{\mathrm{d}},0\}, 
\end{align}
where $\rho \in [0,1]$ is the effective data proportion, the processed data amount 
$A_p(t)=\min\{Q_\mathrm{cmp}(t-1)+A_s(t), \kappa f_\mathrm{cmp}(t)T_{\mathrm{d}}\}$.
Similar to (\ref{Equation:11}), the long-term constraint on $Q_\mathrm{com}(t)$ is 
\begin{align}\label{Equation:14}
\sum_{t=1}^{T}(\rho\kappa f_\mathrm{cmp}(t)-R(t)) \leq0.
\end{align}

\subsection{Problem Formulation }
In this paper, we aim to extend the battery life by minimizing the depth of discharge, i.e., the amount of discharge energy over the time horizon $T$.  We formulate the problem  as 
\begin{align}\label{Equation:15}
&\min_{\{ E_{\mathrm{out}}, f_\mathrm{cmp}, R\}} \sum_{t=1}^{T}E_{\mathrm{out}}(t) \\
\nonumber
&s.t. ~~(\ref{Equation:4}),  (\ref{Equation:11}),  (\ref{Equation:14}).
\end{align}
It requires a holistic optimization of   discharge energy amount, CPU frequency, and transmission rate  to adapt to system dynamics such as energy arrival rate and satellite mobility, subject to real-time and long-term constraints.

\section{Algorithm Design}

If the full information of energy arrival rate and wireless channel state information over the whole time horizon $T$ is known, the problem is a convex optimization, which can be solved offline.  In our scenario, the  energy arrival rate can be calculated precisely. 
However, the wireless channel state information of the satellite-ground link is  hard to predict in a long run by complex prediction methods. If we can predict the wireless channel state information in the current slot,  the problem is a stochastic network optimization problem.
We can adopt  Lyapunov optimization theory to transform the long-term problem into real-time ones, which are convex optimization problems  \cite{6813406}.
However, the channel state information in the current slot is unknown before making decisions. Hence, the scheduler should adapt its decisions based on the results of previous slots.  A typical setting for such online optimization and learning is online convex optimization \cite{boyd2004convex}. 
\subsection{Optimal Pattern-Aware  Benchmark}\label{Sec:OCS}
Although  the channel state information is unavailable, we can exploit the pattern information when the satellite moves around the Earth. For example, the periodical satellite movement generates periodical-like pattern  or trend information, i.e., the satellite-ground connection is periodical as  satellites move around the Earth and satellites show up in  the light and  eclipse alternately.  We design a pattern-aware benchmark to adapt to the different characteristics of the environment and energy harvesting across time slots.
Given the time horizon $T$, we first create a partition which splits the pattern space $\Omega$ into $K=(L)^N$ hypercubes of identical size
$\frac{1}{L}\times\frac{1}{L}\times\cdots\times\frac{1}{L}$. These hypercubes correspond to the different environment and energy harvesting characteristics. Then 
 we define  the time window $\mathcal{W}_k$  contains all the time slots  whose pattern  $\mathbf{c}(t)$ belongs to the $k_\mathrm{th}$ hypercube, i.e.,
\begin{align}\label{Equation:16}
&\mathcal{W}_k=\{t: \mathbf{c}(t) \in \Omega_k\}, \forall k=\{1, \cdots, K\}.
\end{align}
The partitioning of the time horizon captures a general pattern information. 

Then, we introduce an optimal pattern-aware benchmark to measure the performance of pattern-aware algorithms.  We denote the decision variable with $\mathbf{x}=(E_{\mathrm{out}}(t), f_\mathrm{cmp}(t), R(t))$ and the objective function with $f^t(\mathbf{x}(t))$ for simplicity.  Given a sequence of channel state information $\{h(1), \cdots, h(T)\} $ over the time horizon $T$, the optimal pattern-aware benchmark finds $K$  energy scheduling strategies $[\mathbf{x}^*(1), \cdots, \mathbf{x}^*(K)]$ for each time window  $\mathcal{W}_k$ as 
\begin{align}\label{Equation:17}
\mathbf{x}^*(k)= \arg \min_{\mathbf{x}}\sum_{t\in \mathcal{W}_k} f^t(\mathbf{x}),
\end{align}
which is optimal regarding  information only in the respective time window $\mathcal{W}_k$. The optimal pattern-aware  benchmark is unavailable because it requires the full knowledge of the $h(t)$  in each time window $\mathcal{W}_k$. However, it captures  energy harvesting and satellite ground connection patterns.

A performance metric to evaluate the learning performance of online  algorithms is \textit{regret}: the difference between the online algorithm and a benchmark. 
The accumulative regret of an algorithm  with respect to the optimal pattern-aware  benchmark over the whole $T$ time horizon  is  defined as
\begin{align}
Reg_A(T,K)= \sum_{t=1}^T f^t(\mathbf{x}(t))- \sum_{k=1}^K\sum_{t\in \mathcal{W}_k} f^t( \mathbf{x}^*(k)).
\end{align}
Our regret definition is different from the common metrics of static regret and dynamic regret \cite{zinkevich2003online} because all three compare with different benchmarks. 
The static regret compares with an optimal static benchmark which finds the best  fixed strategy  $\mathbf{x}^*$ in hindsight  over the $T$ time horizon. 
The dynamic regret compares with  an optimal dynamic benchmark which  finds the best strategy $\mathbf{x}^*(t)$ for each time slot $t$.  The pattern-aware  regret  is general as it reduces to the static regret when $K=1$ and  to the dynamic regret when $K=T$ and each time slot has unique pattern information. 
Note that the optimal dynamic benchmark has the best performance compared with both the optimal static benchmark and the optimal pattern-aware  benchmark when the underlying system optima is inherently changing \cite{8338087}. We can infer that  the optimal pattern-aware  benchmarks with larger $K$ perform better in our scenario with the dynamics of energy harvesting  and  the uncertainty of the wireless environment.
We aim to is find a sequence $\mathbf{x}(t)$ such that the regret $Reg_A(T,K)$ grows sub-linearly  with respect to $T$.
\subsection{Online Energy Scheduling with No Regret}

In Lyapunov optimization, we introduce  virtual queues to control the long-term constraint violations. Hence, we introduce a low-complexity virtual queue based algorithm that addresses  our  online convex optimization problem with long-term constraints  \cite{yu2016low}. 
To describe the algorithm more clear, we first introduce some notations. 
We denote the gradient of the objective function as $\nabla f^t(\mathbf{x}(t))$. 
The  convex set $\mathcal{X}_0$ is defined by instantaneous constraints, i.e.,
\begin{align}
&\mathcal{X}_0=\{\mathbf{x}: 0\leq f_\mathrm{cmp}(t) \leq f_\mathrm{cmp}^{\max}(t), \\
   \nonumber
  & 0\leq R(t) \leq R^{\max}I_c(t),   E_\mathrm{out} \geq 0\}, \forall t.
\end{align}
 We denote the long-term constraint functions as $\mathbf{g}(\mathbf{x}(t))=(g_1(x(t)), g_2(x(t)))$ where $g_1(x(t))=Df_\mathrm{sen}(t)-\kappa f_\mathrm{cmp}(t)$, $g_2(x(t))=\rho\kappa f_\mathrm{cmp}(t)-R(t)$. Let $ \tilde{\mathbf{g}}(\mathbf{x}(t))=\gamma \mathbf{g}(\mathbf{x}(t))$, where $\gamma$ is a positive constant. 
We introduce a virtual queue vector $\mathbf{Q}^V(t)=(Q_1^V(t), Q_2^V(t))$  for the long-term constraint vector $\mathbf{g}(\mathbf{x}(t))$.
 The algorithm first chooses arbitrary feasible $\mathbf{x}(1) \in \mathcal{X}_0$ and  then chooses $\mathbf{x}(t+1)$ that solves the following problem 
 \begin{align} \label{eq23}
&\min_{\mathbf{x} \in \mathcal{X}_0}\{[\nabla f^t(\mathbf{x}(t))]^\mathbf{T}[\mathbf{x}-\mathbf{x}(t)]\\
\nonumber
&+[\mathbf{Q}^V(t)+\tilde{\mathbf{g}}(\mathbf{x}(t))]^\mathbf{T} \tilde{\mathbf{g}}(\mathbf{x})+\alpha ||\mathbf{x}-\mathbf{x}(t)||^2\},
\end{align}
where $\alpha$ is a positive constant and the virtual queues 
 \begin{align}\label{Equation:21}
Q_i^V(t)=\max\{-\tilde{g}_i(\mathbf{x}(t)), Q_i^V(t-1)+\tilde{g}_i(\mathbf{x}(t))\}.
\end{align}
Note that the virtual queues are different from the data queues and the energy queue in the system model. The value  of $\mathbf{Q}^V(t)$ is a queue backlog  of constraint violations. By introducing the virtual queue vector $\mathbf{Q}^V(t)$, 
the algorithm transforms the long-term constraints into queue length variation and solve the problem in (\ref{eq23}) in each time slot $t$.

 \textbf{Linear constraints.} We observe that our $\mathbf{g}(\mathbf{x}(t))$ is affine i.e,  $\mathbf{g}(\mathbf{x}(t))= \mathbf{A}\mathbf{x}-\mathbf{b}$, where $\mathbf{A}= [0, -\kappa, 0 ;  0, \rho \kappa,  -1]$ and $\mathbf{b}=[-Df_s(t), 0]^\mathbf{T}$, then the update of $\mathbf{x}(t+1)$ can be solved by a projection onto a convex set as in Lemma \ref{lem1}.
 \begin{lemma} \cite{yu2016low} \label{lem1}
 If  $\mathbf{g}(\mathbf{x}(t))$ is affine,  i.e,  $\mathbf{g}(\mathbf{x}(t))= \mathbf{A}\mathbf{x}-\mathbf{b}$ for some matrix $\mathbf{A}$ and vector $\mathbf{b}$, then we can update $\mathbf{x}(t+1)$ as
 \begin{align}\label{Equation:22}
\mathbf{x}(t+1)=\mathcal{P}_{x\in \mathcal{X}_0}\left[\mathbf{x}-\left(\mathbf{x}(t)-\frac{1}{2\alpha}\mathbf{d}(t)\right)\right],
\end{align}
 where
  \begin{align}
\mathbf{d}(t)=\nabla f^t(\mathbf{x}(t))+\sum_{i=1}^{2}[Q_i^V(t)+\tilde{g}_i(\mathbf{x}(t))]\nabla\tilde{g}_i(\mathbf{x}(t)).
\end{align}
  \end{lemma}
 \begin{algorithm} [t]
\caption{Pattern-Aware Online Energy Scheduling}
\label{algorithm: al3}
\begin{algorithmic}[1]
\REQUIRE ~~ \\
Constant parameter  $\alpha$, $\gamma, \beta$,  time slot duration $T_d$, 
energy related parameters: $E_\mathrm{total}(t)$, $E_a(t)$,  $E_{\max}$, 
sensing related parameters $D$,  $P_\mathrm{sen}(t), f_\mathrm{sen}(t)$, 
computing related parameters $\kappa, \rho, f_\mathrm{cmp}^{\max}$, 
communication related parameter $B, N_0, I_c(t), R^{\max}$.
\ENSURE ~~ \\
Energy scheduling strategies $\mathbf{x}(t), t=\{1, \cdots, T\}$.
\STATE  Initialization: $Q_\mathrm{com}(0), Q_\mathrm{cmp}(0), E(0), Q_i^V(0)$ as 0. \\
\FOR {Each time slot $t= 1, \cdots, T$}
\STATE Identify time window $\mathcal{W}_k \ni t$
\IF {$t=t_k^1$ for $\mathcal{W}_k$}
\STATE Choose arbitrary $\mathbf{x}(t_k^1) \in \mathcal{X}_0$.
\ELSE
\STATE Choose $\mathbf{x}(t_k^\tau)$ by solving the projection 
 \begin{align}\label{Equation:24}
\mathcal{P}_{x\in \mathcal{X}_0}\left[\mathbf{x}-\left(\mathbf{x}(t_k^{\tau-1})-\frac{1}{2\alpha}\mathbf{d}(t_k^{\tau-1})\right)\right].
\end{align}
\ENDIF
\STATE Observe actual $h(t)$.
\STATE Update the virtual queue vector $\mathbf{Q}^V$ via
 \begin{align}\label{Equation:25}
Q_i^V(t)=\max\{-\tilde{g}_i(\mathbf{x}(t)), Q_i^V(t_k^{\tau-1})+\tilde{g}_i(\mathbf{x}(t))\}.
\end{align}
\STATE Update the energy queue $E(t)$, data queues $Q_\mathrm{com}(t)$ and $Q_\mathrm{cmp}(t)$ according to (\ref{Equation:2}), (\ref{Equation:9}), (\ref{Equation:12}).
\ENDFOR
\end{algorithmic}
\end{algorithm}

 We design our pattern-aware online energy scheduling algorithm as in Algorithm  \ref{algorithm: al3}. 
Our algorithm exploits the periodical energy harvesting and connection information then it learns to make decisions based on the historical information in each time window $\mathcal{W}_k$.  We use  $t_k^1, t_k^\tau$ to denote the first  time slot and  the $\tau_\mathrm{th}$ time slot in each time window $\mathcal{W}_k$, respectively. At $t= t_k^1$ of $\mathcal{W}_k$, our algorithm chooses any feasible energy scheduling strategy because it has no previous information to rely on. At $t= t_k^\tau$ of $\mathcal{W}_k$, it updates the energy scheduling strategy by solving the projection in (\ref{Equation:24}). After receiving the channel state information of the current slot, it updates virtual queues, the energy queue,  and  data queues  via (\ref{Equation:25}), (\ref{Equation:2}), (\ref{Equation:9}), and  (\ref{Equation:12}), respectively. Our algorithm is simple as it either selects a strategy arbitrarily or  performs an easy gradient descent,  incurring neglected energy consumption.
 
 \textbf{Performance Analysis}.
 We analyze the objective function and constraints then make a mild assumption for the theoretical proof. The  feasible region $\mathcal{X}_0$ is bounded in our problem, then there exists a constant $G_1$ such that  $||\mathbf{x}-\mathbf{y}||\leq G_1$  and a constant  $G_2$ such that $||\mathbf{g}(\mathbf{x})||\leq G_2$, for all  $\mathbf{x}, \mathbf{y}\in  \mathcal{X}_0$. The objective function has bounded gradient on the feasible region $\mathcal{X}_0$, i.e., $||\nabla f^t(\mathbf{x}(t))||=1, \forall \mathbf{x}, t$. There exists a constant $\beta$ such that $||\mathbf{g}(\mathbf{x})-\mathbf{g}(\mathbf{y})||\leq \beta||\mathbf{x}-\mathbf{y}||$ for all  $\mathbf{x}, \mathbf{y}\in  \mathcal{X}_0$. 
 \begin{assumption}
Assume that there exist $\epsilon$ and $\hat{\mathbf{x}}\in  \mathcal{X}_0$ such that $\mathbf{g}_i(\hat{\mathbf{x}})\leq \epsilon$ for all $i=1, 2$.
  \end{assumption}

The slater condition is mild for convex optimization. We characterize the regret and constraint violations for Algorithm \ref{algorithm: al3} through the  analysis of  a drift plus penalty expression following the idea of \cite{yu2016low} as in Lemma 2.

\begin{lemma}
Consider online convex optimization with long-term constraints that satisfy Assumption 1. Let $\mathbf{x}^* \in \mathcal{X}_0$ be any fixed solution that satisfies $\mathbf{g}(\mathbf{x}^*)$, e.g., $\mathbf{x}^*=\arg \min_{\mathbf{x} \in \mathcal{X}_0} \sum_{t=1}^T f^t(\mathbf{x})$.
 Let $\gamma >0, \eta>0$ be arbitrary.
 
 1. If $\alpha\geq \frac{1}{2}(\gamma^2\beta^2+\eta)$ in Algorithm \ref{algorithm: al3}, then for all $T\geq1$, Algorithm 1 in \cite{yu2016low} has 
 \begin{align}
 \sum_{t=1}^T f( \mathbf{x}(t))\leq \sum_{t=1}^T f( \mathbf{x}^*)
 +\alpha ||(( \mathbf{x}^*- \mathbf{x}(1))||^2+\frac{T}{2\eta}.
 \end{align}
  2. For all $T\geq 1$, the constraint violations of Algorithm 1 in \cite{yu2016low} is bounded as 
  \begin{align}
 \sum_{t=1}^T{g}_i(\mathbf{x}(t))\leq 2G_2+\frac{\alpha G_1^2+G_1}{\gamma^2\epsilon}+\frac{2 G_2^2}{\epsilon}, \forall i=1, 2.
  \end{align}
\end{lemma}
We extend theoretical results  of  \cite{yu2016low}   to fit the context of our pattern-aware scenario as follows.
 \begin{theorem} 
If $\gamma = T^{1/4}, \eta= \sqrt{T}$ and $\alpha= \frac{1}{2}(\beta^2+1)\sqrt{T}$, for all $T\geq 1$, our algorithm has sublinear regret as 
 \begin{align}\label{Equation:25}
 \sum_{k=1}^K\sum_{t \in \mathcal{W}_k} f( \mathbf{x}(t))\leq \sum_{k=1}^K\sum_{t \in \mathcal{W}_k} f( \mathbf{x}^*(k)) +O(\sqrt{T}).
  \end{align} 
  For all $T\geq 1$, the constraint violations are bounded as
  \begin{align} \label{Equation:26}
    \sum_{k=1}^K \sum_{t \in \mathcal{W}_k} g_i(\mathbf{x}(t))\leq 2KG_2+\\
    \nonumber
    \frac{K(\frac{1}{2}(\beta^2+1) G_1^2+G_1+2 G_2^2)}{\epsilon}+\frac{KG_1}{\epsilon\sqrt{T}}.
  \end{align}
 \end{theorem}
 \begin{proof}
We adopt  Lemma 2  for each time window $\mathcal{W}_k$, separately. 
 If $\alpha\geq \frac{1}{2}(\gamma^2\beta^2+\eta)$ in Algorithm \ref{algorithm: al3}, then for all $|\mathcal{W}_k|\geq1$, we have
 \begin{align}
 \sum_{t \in \mathcal{W}_k} f( \mathbf{x}(t))\leq \sum_{t \in \mathcal{W}_k} f( \mathbf{x}^*)
 +\alpha ||(( \mathbf{x}^*(k)- \mathbf{x}(t_k^1))||^2+\frac{|\mathcal{W}_k|}{2\eta}.
 \end{align}
 For the whole time horizon $T$, we have
 \begin{align}
 \sum_{k=1}^K\sum_{t \in \mathcal{W}_k} f( \mathbf{x}(t))\leq \sum_{k=1}^K\sum_{t \in \mathcal{W}_k} f( \mathbf{x}^*)\\
 \nonumber
 +\sum_{k=1}^K\alpha ||(( \mathbf{x}^*(k)- \mathbf{x}(t_k^1))||^2+\sum_{k=1}^K\frac{|\mathcal{W}_k|}{2\eta}.
 \end{align} 
If $\gamma = T^{1/4}, \eta= \sqrt{T}$ and $\alpha= \frac{1}{2}(\beta^2+1)\sqrt{T}$, we have (\ref{Equation:25}).
  
  The  constraint violation bound  is irrelevant to the time horizon. Thus, for each time window $\mathcal{W}_k$, we have 
    \begin{align}
 \sum_{t \in \mathcal{W}_k} {g}_i(\mathbf{x}(t))\leq 2G_2+\frac{\alpha G_1^2+G_1}{\gamma^2\epsilon}+\frac{2 G_2^2}{\epsilon}, \forall i=1, 2.
  \end{align}
   For the whole time horizon $T$, we have
 \begin{align}
   \sum_{k=1}^K\sum_{t \in \mathcal{W}_k} {g}_i(\mathbf{x}(t))\leq  K(2G_2+\frac{\alpha G_1^2+G_1}{\gamma^2\epsilon}+\frac{2 G_2^2}{\epsilon}), \forall i=1, 2.
  \end{align}
If $\gamma = T^{1/4}, \eta= \sqrt{T}$ and $\alpha= \frac{1}{2}(\beta^2+1)\sqrt{T}$, we have (\ref{Equation:26}).
 \end{proof}
  Theorem 1  implies that if we choose $\gamma = T^{1/4}$ and $\alpha= \frac{1}{2}(\beta^2+1)\sqrt{T}$ in Algorithm \ref{algorithm: al3} then we can achieve $O(\sqrt{T})$ regret and $O(\frac{1}{\sqrt{T}})$ constraint violation. Note that the theoretical results are based on the parameters $G_1, G_2, \epsilon$ while it  only requires $\beta$ to implement Algorithm \ref{algorithm: al3}, which is known since the constraint functions $\mathbf{g}(\mathbf{x}(t))$ do not change.
  
\section{Simulation Results}
\subsection{Simulation Setting}

\begin{figure*}[htbp]
\centering
\begin{minipage}[t]{0.245\linewidth}
\centering
{\includegraphics[width=1\linewidth]{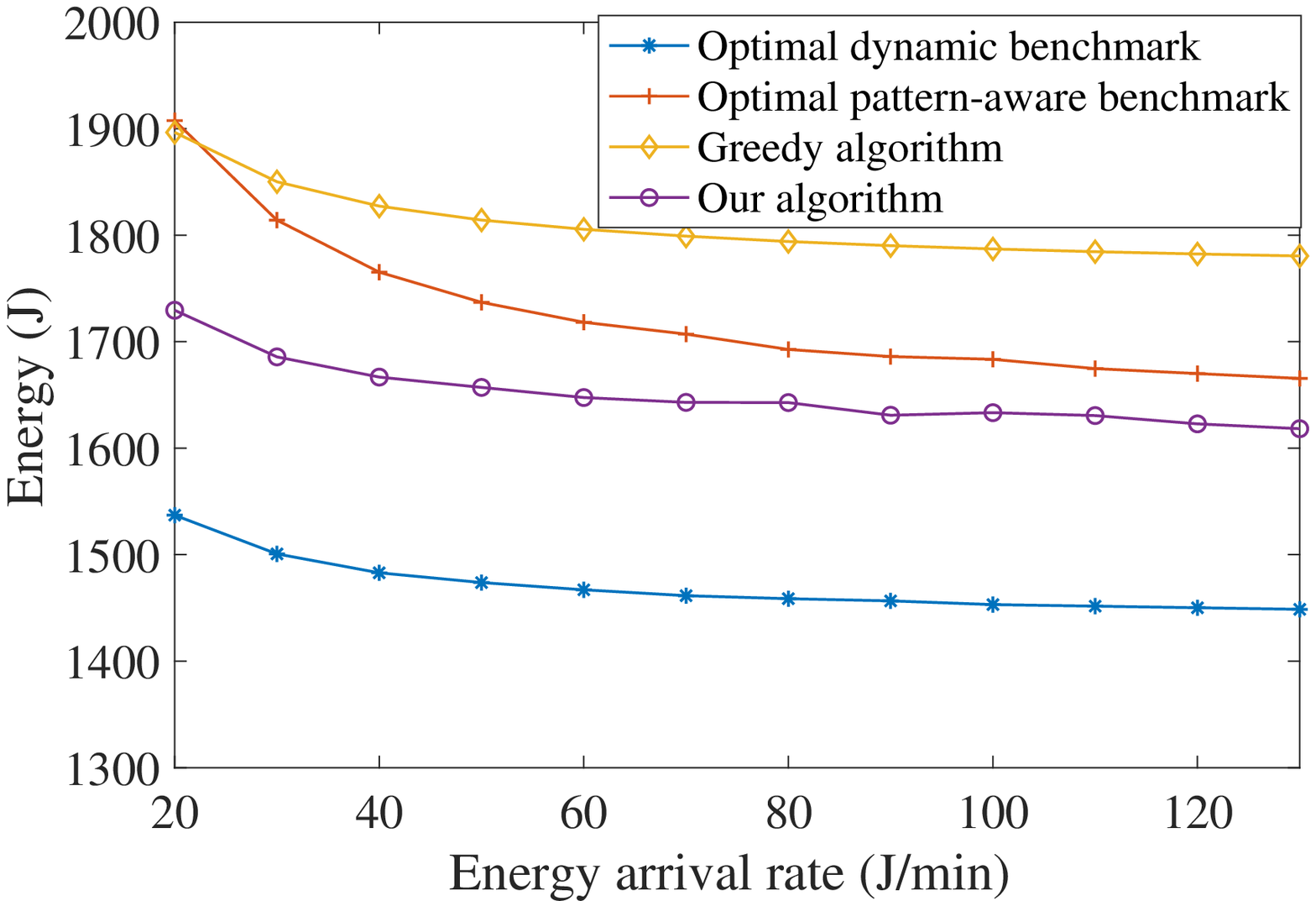}}
\caption{Impact of harvested energy density.}
\label{Fig: Density}
\end{minipage}
\hfill
\begin{minipage}[t]{0.245\linewidth}
\centering
{\includegraphics[width=1\linewidth]{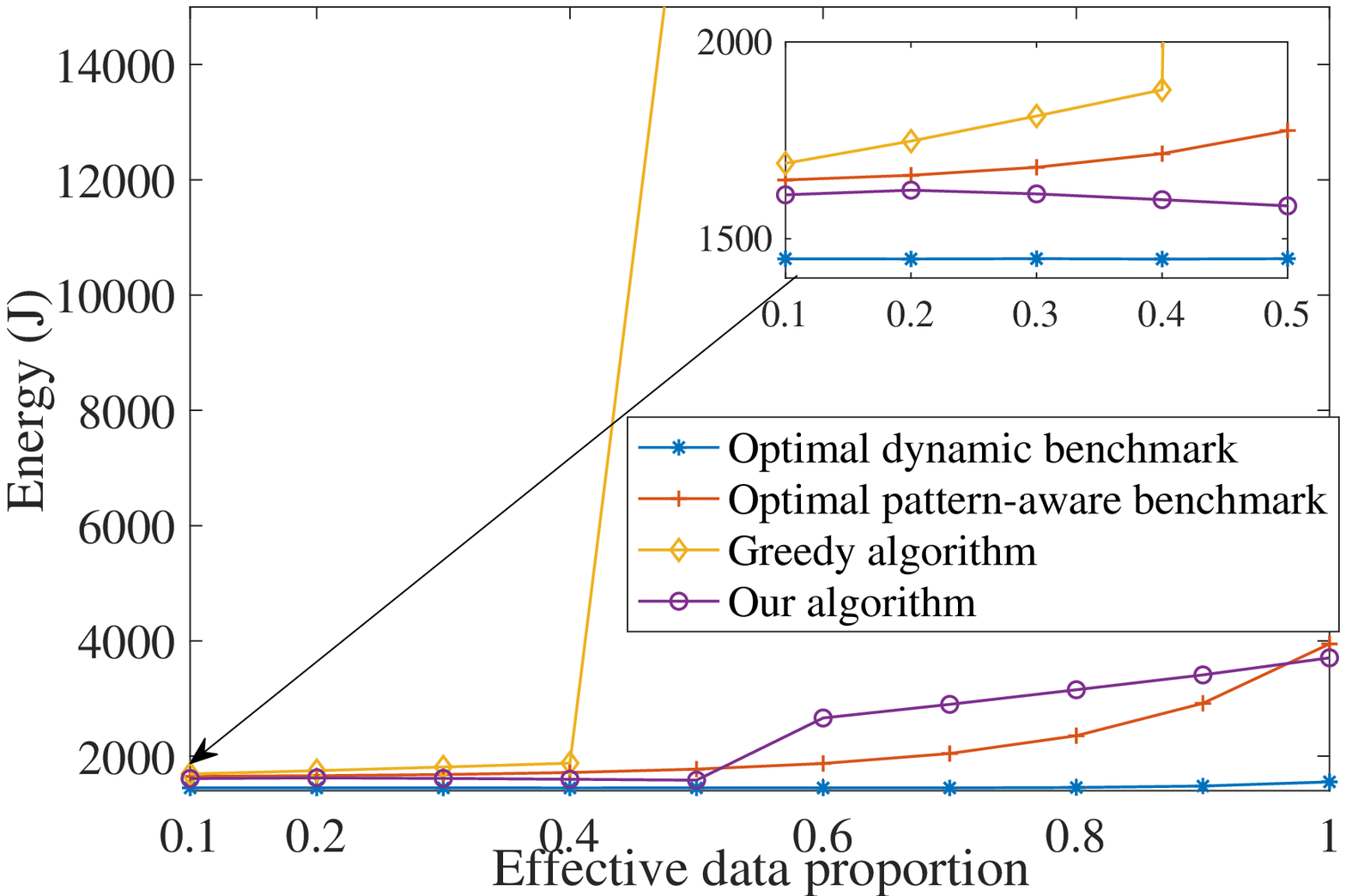}}
\caption{Impact of  effective data proportion.}
\label{Fig: Rho}
\end{minipage}
\hfill
\begin{minipage}[t]{0.245\linewidth}
\centering
{\includegraphics[width=1\linewidth]{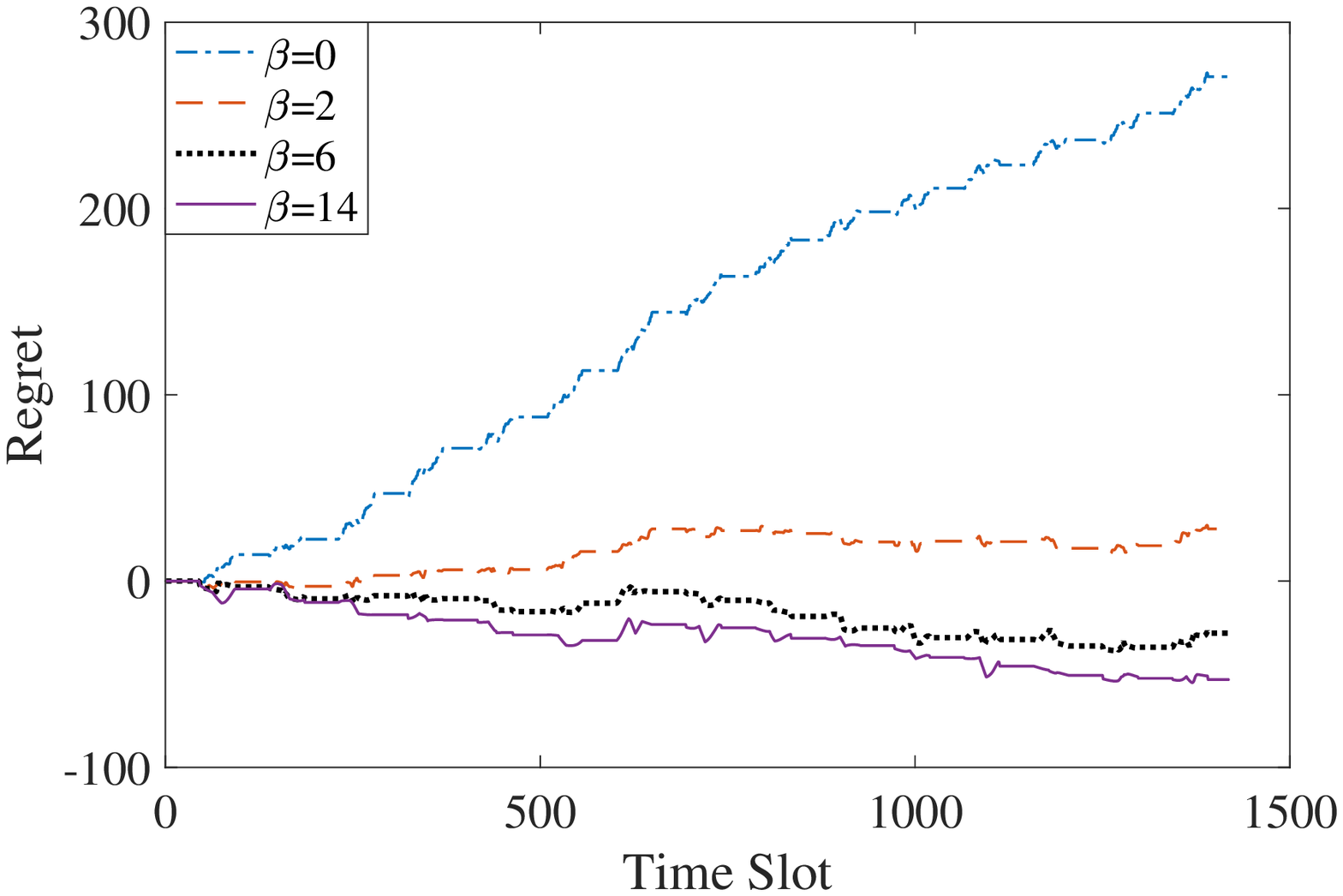}}
\caption{Algorithm regret.}
 \label{Fig:Regret}
\end{minipage}
\begin{minipage}[t]{0.245\linewidth}
\centering
{\includegraphics[width=1\linewidth]{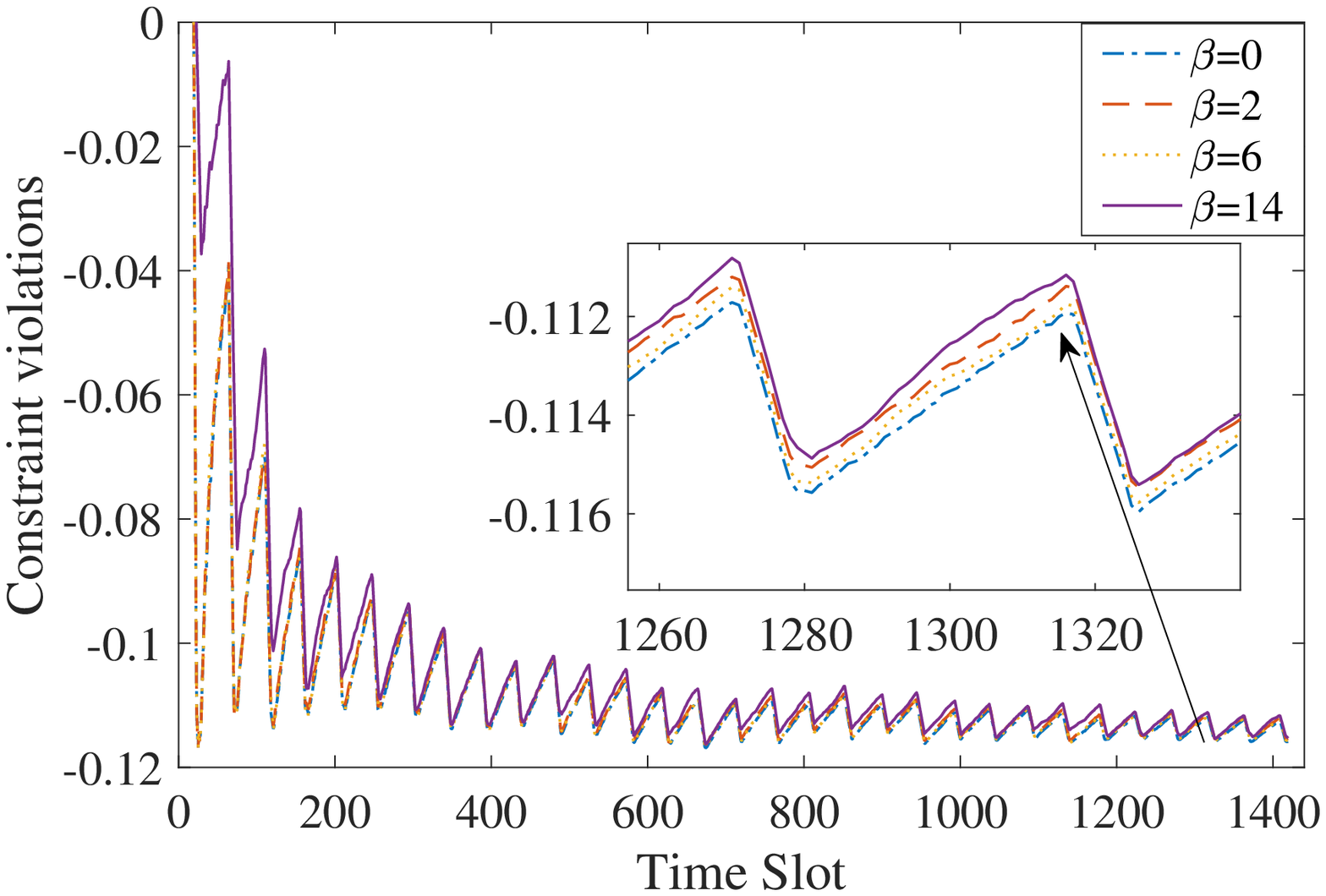}}
\caption{Constraint violations.}
 \label{Fig:Violation}
\end{minipage}

\end{figure*}

We simulate a scenario where an LEO satellite orbits the Earth and captures Earth images according to the mission demand.   We select a 6am Sun-synchronous orbit and the orbit  altitude  is set  to  550km. We simulate two ground stations,  one in the light and the other in the eclipse.  
The satellite harvests energy as it orbits the Earth.  We calculate the  energy harvesting rate $r_e(t)$ (in J/min)  by $r_e(t)=E_s^{\max}\cos(\theta_t)$, where $E_d$ is the maximal harvest  power and $\theta_t$ is the angle between the Sun and the normal vector  of  the solar panel mentioned in Section \ref{sec:Energy Harvesting and Storage Model}. We set $E_d=30$ J/min and $\theta_t$ varies between $[0, \pi/2]$. We set  the main simulation parameters according to \cite{mysatellite} and \cite{9155418}. The energy consumption of other subsystems $E_a=0$ for simplicity. Battery capacity $E_{\max}=10800$ J, and the battery is full of charge at the beginning of the mission. The data size of each frame is set  to 60 Mbits. The sensing power is set  to 2 J/min. Sensing frame rate $f_\mathrm{sen}(t)$ is randomly generated in  [0, 4].  The data to computation ratio $\kappa$ is set  to 0.1 bit/cycle.
The effective data proportion $\rho=0.25$.
Maximal CPU frequency $f_p^{\max}=4$  GHz. Bandwidth $B=80$ MHz.
Parameter $h/N_0$ is randomly generated  in $[15,20]$.
Maximal transmission bitrate $R_{\max}$ is set  to $500$ Mbit/s.
Algorithm Parameters are set  to  ${\beta}=14$  and ${\gamma}=1440^{1/4}$.
In the simulation, we schedule the energy every one minute over 24 hours. Hence, the whole time horizon  $T=1440$. We divide the whole time horizon into $K=4$ time windows.

\textbf{ Baseline Approaches.} We compare our algorithm with three baselines as follows.

 1) \textit{Optimal dynamic benchmark}. The scheduler  knows channel state information $\{h(1), \cdots, h(T)\} $ over the time horizon $T$ and  finds  optimal energy scheduling strategies $[\mathbf{x}^*(1), \cdots, \mathbf{x}^*(T)]$.  This benchmark provides a performance upper bound for the evaluated algorithms. We achieve the optimal dynamic benchmark directly by the convex optimization tool, CVX \cite{mycvx}. 
 
 2) \textit{Optimal pattern-aware benchmark}. The scheduler  knows channel state information $\{h(1), \cdots, h(T)\} $ over the time horizon $T$ and finds $K$  energy scheduling strategies $[\mathbf{x}^*(1), \cdots, \mathbf{x}^*(K)]$ for each time window  $\mathcal{W}_k$ as defined in (\ref{Equation:17}),
which is optimal only  regarding energy arrival rate in each time window $\mathcal{W}_k$. 

 3) \textit{Greedy algorithm}. The scheduler  knows channel state information of the current time slot $h(t)$ and allocate minimal energy to sensing, computing, and communication to satisfy the constraints in (\ref{Equation:11}) and  (\ref{Equation:14}) in each time slot $t$.

\subsection{Simulation Results}
 In Figure~\ref{Fig: Density}, the depth of discharge of our algorithm is less than both  the optimal pattern-aware  benchmark and the greedy algorithm  over the time horizon  when the energy harvesting capability varies in the range of [20, 130] J.   
 The optimal dynamic benchmark achieves the minimal depth of  discharge. With the full information,  it allocates the least energy only for sensing and no energy for computation and communication in the eclipse. And it allocates as much energy (directly from the solar panels) as possible to finish the accumulative workloads in the light. As expected, the optimal pattern-aware benchmark performs worse than the optimal dynamic benchmark because  it only optimizes  for each time window without global scheduling and gets $4$ fixed solutions for each orbit period, which cannot track the system dynamics.  Our algorithm can adapt to the system dynamics by learning from historical information. 
In our algorithm, the depth of discharge decreases slowly with the harvested energy density. It implies that our algorithm is robust to the variation of maximal harvesting  energy and  can work well with limited  harvesting capability.  The optimal dynamic benchmark is also robust to the variation of maximal harvesting  energy. 

 In Figure~\ref{Fig: Rho},  we aim to evaluate how the effective data proportion impacts our algorithm. We set the parameter $\rho$  in [0.1, 1] and record the sum of the amount of discharge  energy over the whole time horizon $T$. We can observe that our algorithm performs better than both the optimal pattern-aware  benchmark and the greedy algorithm when $\rho \leq 0.5$ and its amount of  discharge is larger than the optimal pattern-aware  benchmark when $\rho > 0.5$. The effective data proportion influences the data amount to be transmitted to the ground. The transmission rate will increase with  $\rho$, and  the energy increases exponentially with transmission rate, which results in deeper battery discharge  in the downlink session of the eclipse.    The greedy algorithm  performance deteriorates most sharply with $\rho$. The two optimal  benchmarks are robust to the variation of $\rho$ because they can optimize in larger timescale and keep the workload  to be processed or transmitted in the light to reduce the depth of discharge in the eclipse.

 In Figure \ref{Fig:Regret} and Figure \ref{Fig:Violation}, our algorithm has no regret against the optimal pattern-aware  benchmark and  the constraint violation asymptotically  approaches zero as proved in Theorem 1.  We plot the accumulative regret and constraint violations with time under different values of parameter $\beta$.  Despite the energy harvesting dynamics and the lack of channel gain information, our algorithm can achieve a near-optimal solution.  When $\beta=0$ (not a validate value), the accumulative regret grows linearly with the time horizon because it leads to a large step size $\frac{1}{2\alpha}$  in (\ref{Equation:22}) and the learning performance is poor. The regret gets smaller with the increase of $\beta$.  Increasing $\beta$ will result in smaller step size $\frac{1}{2\alpha}$, the algorithm will fail to adapt  to the system dynamics, thus incurring many constraint violations  as in Figure~\ref{Fig:Violation}.
 We find  that our algorithm can achieve the minimal depth of  discharge  without constraint violations when $\beta=14$.


\section{Conclusion}
In this paper, we extend the  batteries'  life of LEO satellites by minimizing the depth of discharge for Earth observation missions.  We propose  a pattern-aware online  energy  scheduling algorithm that decides the depth of discharge, CPU frequency, and transmission rate under lack of wireless channel state information.  Our algorithm has a theoretical guarantee of  learning loss  and  reduces  the depth of discharge significantly.   Even our solution cannot  fundamentally extend  batteries' life like inventing new battery materials, it can work well together with those advanced battery technologies. In  future work, we are interested to optimize the depth of discharge for large-scale constellations with satellites collaboration.
\section{Acknowledgement}
This work was supported by the National Natural Science Foundation of China (61921003, 61922017, and 62032003).

\bibliographystyle{IEEEtran}
\bibliography{IEEEabrv,reference210903}

\end{document}